\documentclass[10pt,reqno]{amsart}
\usepackage{natbib}
\usepackage{amsaddr}
\usepackage{etoolbox}
\patchcmd{\section}{\scshape}{\bfseries\scshape}{}{}
\makeatletter
\renewcommand{\@secnumfont}{\bfseries}
\makeatother

\usepackage[utf8]{inputenc} 
\usepackage[T1]{fontenc}    
\usepackage{hyperref}       
\usepackage{url}            
\usepackage{booktabs}       
\usepackage{amsfonts}       
\usepackage{nicefrac}       
\usepackage{microtype}      
\usepackage[a4paper,margin=3cm]{geometry}

\usepackage{amsmath}
\usepackage{color}
\usepackage{enumitem}
\usepackage{amssymb}
\usepackage{amsthm}
\usepackage{graphicx}
\usepackage{dsfont}

\DeclareMathOperator{\Id}{Id}

\definecolor{NavyBlue}{rgb}{0.1,0.1,0.6}

\def\bfone{{\boldsymbol 1}}

\newcommand{\diff}{\mathrm{d}}

\newcommand{\N}{\mathbb{N}}

\newcommand{\E}{\mathbb{E}}
\newcommand{\R}{\mathbb{R}}

\renewcommand{\P}{\mathds{P}}

\newcommand{\cP}{\mathcal{P}}

\newcommand{\cF}{\mathcal{F}}

\renewcommand{\leq}{\leqslant}
\renewcommand{\geq}{\geqslant}

\newtheorem{remark}{Remark}

\newtheorem{definition}{Definition}

\newtheorem{proposition}{Proposition}
\newtheorem{thm}{Theorem}

\makeatletter
\g@addto@macro{\endabstract}{\@setabstract}
\newcommand{\authorfootnotes}{\renewcommand\thefootnote{\@fnsymbol\c@footnote}}%
\makeatother

\title[A Continuized View on Nesterov Acceleration]{A Continuized View on Nesterov Acceleration}

\begin{document}
\begin{center}
	\LARGE 
	A Continuized View on Nesterov Acceleration \par \bigskip
	
	\normalsize
	\authorfootnotes
	Raphaël Berthier\textsuperscript{1}, Francis Bach\textsuperscript{1},
	Nicolas Flammarion\textsuperscript{2}, Pierre Gaillard\textsuperscript{3} and
	Adrien Taylor\textsuperscript{1} \par \bigskip
	
	\textsuperscript{1}Inria - Département d’informatique de l’ENS \\
	PSL Research University, Paris, France
	\smallskip \par
	\textsuperscript{2}School of Computer and Communication Sciences \\
	Ecole Polytechnique F\'ed\'erale de Lausanne\smallskip \par
		\textsuperscript{3}Univ. Grenoble Alpes, Inria, CNRS, Grenoble INP, LJK, 38000 Grenoble, France \par \bigskip
\end{center}

\begin{abstract}
We introduce the ``continuized'' Nesterov acceleration, a close variant of Nesterov acceleration whose variables are indexed by a continuous time parameter. The two variables continuously mix following a linear ordinary differential equation and take gradient steps at random times. This continuized variant benefits from the best of the continuous and the discrete frameworks: as a continuous process, one can use differential calculus to analyze convergence and obtain analytical expressions for the parameters; but a discretization of the continuized process can be computed exactly with convergence rates similar to those of Nesterov original acceleration. We show that the discretization has the same structure as Nesterov acceleration, but with random parameters. 
\end{abstract}

\section{Introduction}

In the last decades, the emergence of numerous applications in statistics, machine learning and signal processing has led to a renewed interest in first-order optimization methods \citep{bottou2018optimization}. They enjoy a low computational complexity necessary to the analysis of large datasets. The performance of first-order methods was largely improved thanks to acceleration techniques \citep[see the review by][and the many references therein]{daspremont2021acceleration}, starting with the seminal work of \cite{nesterov1983method}.

Let $f:\R^d \to \R$ be a convex and differentiable function, minimized at $x_* \in \R^d$. We assume throughout the paper that $f$ is $L$-smooth, i.e.,
\begin{align*}
\forall x,y \in \R^d, \qquad f(y) \leq f(x) + \langle \nabla f(x), y-x \rangle + \frac{L}{2} \Vert y-x \Vert^2 \, .
\end{align*}
In addition, we sometimes assume that $f$ is $\mu$-strongly convex for some $\mu > 0$, i.e., 
\begin{align*}
\forall x,y \in \R^d, \qquad f(y) \geq f(x) + \langle \nabla f(x), y-x \rangle + \frac{\mu}{2} \Vert y-x \Vert^2 \, .
\end{align*}
For the problem of minimizing $f$, gradient descent is well-known to achieve a rate $f(x_k)-f(x_*) = O(k^{-1})$ in the smooth case, and a rate $f(x_k) - f(x_*) = O((1-\mu/L)^k)$ in the smooth and strongly convex case. In both cases, Nesterov introduced an alternative method with essentially the same running-time complexity, that achieves faster rates; it converges at the rate $O(k^{-2})$ in the smooth convex case and at the rate $O((1-\sqrt{\mu/L})^k)$ in the smooth and strongly convex case~\citep{nesterov2003introductory}. These rates are then optimal among all methods that access gradients and linearly combine them~\citep{nesterov2003introductory,nemirovskij1983problem}.

Nesterov acceleration introduces several sequences of iterates---two or three, depending on the formulations---and relies on a clever blend of gradient steps and mixing steps between the iterates. Many works contributed to interpret and motivate the precise structure of the iteration that lead to the success of the method, see for instance \citep{bubeck2015geometric,flammarion2015averaging, arjevani2015lower,kim2016optimized,allen2014linear}.
A large number of these works found useful to study continuous time equivalents of Nesterov acceleration, obtained by taking the limit when stepsizes vanish, or from a variational framework. The continuous time index~$t$ of the limit allowed to use differential calculus to study the convergence of these equivalents. For examples of studies that use continuous time, see \citep{su2014differential,krichene2015accelerated,wilson2016lyapunov,wibisono2016variational,betancourt2018symplectic,diakonikolas2019approximate,shi2018understanding,shi2019acceleration,attouch2018fast,attouch2019rate,zhang2018direct,siegel2019accelerated,muehlebach2019dynamical,sanz2020connections}.

In this paper, we propose another way to obtain a continuous time equivalent of Nesterov acceleration, that we call the \emph{continuized} version of Nesterov acceleration, that does not have vanishing stepsizes. It is built by considering two variables $x_t, z_t \in \R^d$, $t \in \R_{\geq 0}$, that continuously mix following a linear ordinary differential equation (ODE), and that take gradient steps at random times $T_1, T_2, T_3, \dots$. Thus, in this modeling, mixing and gradient steps alternate randomly. 

Thanks to the continuous index $t$ and some stochastic calculus, one can differentiate averaged quantities (expectations) with respect to~$t$. In particular, this leads to simple analytical expressions for the optimal parameters as functions of $t$, while the optimal parameters of Nesterov accelerations are defined by recurrence relations that are complicated to solve. 

The discretization $\tilde{x}_k = x_{T_k}, \tilde{z}_k = z_{T_k}, k\in \N$, of the continuized process can be computed directly and exactly: the result is a recursion of the same form as Nesterov iteration, but with randomized parameters, that performs similarly to Nesterov original deterministic version both in theory and in simulations. 

There are particular situations where Nesterov acceleration can not be implemented and the continuized acceleration can. First, a major advantage of the continuized acceleration over Nesterov acceleration is that the parameters of the algorithm depend only on time $t \in \R_{\geq 0}$, and not on the number of past gradient steps $k$. This is useful in distributed implementations, where the total number of gradient steps taken in the network may not be known to a particular node. Second, the continuized modeling can be relevant when gradient steps arrive at random times, as in asynchronous parallel computing for instance. Gossip algorithms represent another example where both features are present: the total number of past communications in the network at a given time is unknown to all nodes, and communication between nodes occur at random times. This motivated \cite{even2020asynchrony} to consider a similar continuized procedure, for communication steps instead of gradient steps, in order to accelerate gossip algorithms; their work is the source of inspiration for the present paper. 

Beyond these particular situations, the continuized acceleration should be seen as a close approximation to Nesterov acceleration, that features both an insightful and convenient expression as a continuous time process and a direct implementation as a discrete iteration. We thus hope to contribute to the understanding of Nesterov acceleration. We believe that the continuized framework can be adapted to various settings and extensions of Nesterov acceleration; as an illustration of this statement, we study how the continuized acceleration behaves in the presence of additive noise on the gradients. 

\medskip
\noindent
\textbf{Notations.}
The index $k$ always denotes a non-negative integer, while the indices $t,s$ always denote non-negative reals. 
\medskip

\noindent
\textbf{Structure of the paper.}
In Section \ref{sec:reminders}, we recall gradient descent and Nesterov acceleration, its choice of parameters, and its convergence rate as a function of the number of iterations $k$. In Section~\ref{sec:continuized}, we introduce our continuized variant of Nesterov acceleration, its choice of parameters and its convergence rate as functions of $t$. In Section \ref{sec:implementation}, we show that the discretization of the continuized acceleration leads to an iteration of the same structure as Nesterov acceleration, with random parameters. We give the expressions for the parameters and the convergence rate in terms of the number of iterations $k$. Finally, in Section~\ref{sec:robustness}, we study the robustness of the continuized acceleration to additive noise.

\section{Reminders on gradient descent and Nesterov acceleration}
\label{sec:reminders}

For the sake of comparison, let us first recall classical results of convex optimization. Consider the iterates of gradient descent with stepsize $\gamma$,
\begin{equation*}
x_{k+1} = x_k - \gamma \nabla f(x_k) \, .
\end{equation*}
We have the following convergence of the function values $f(x_k)$, depending on whether the function~$f$ is \eqref{it:gd-cvx} convex, or \eqref{it:gd-str-cvx} strongly convex. 
\begin{thm}[Convergence of gradient descent]
	Choose the stepsize $\gamma = 1/L$. 
	\begin{enumerate}
		\item\label{it:gd-cvx}Then \begin{align*}
		f(x_k)-f(x_*) \leq \frac{2L \Vert x_0 - x_* \Vert^2}{k+4} \, . 
		\end{align*}
		\item\label{it:gd-str-cvx} Assume further that $f$ is $\mu$-strongly convex, $\mu > 0$. Then 
		\begin{align*}
		f(x_k)-f(x_*) \leq \frac{L}{2} \left(1-\frac{\mu}{L}\right)^k \Vert x_0 - x_* \Vert^2 \, . 
		\end{align*}
	\end{enumerate}
\end{thm}
These results (or similar bounds) can be found at many places in the literature; for instance the first bound is in \citep[Corollary 2.1.2]{nesterov2003introductory} and the second bound is a simple consequence of \citep[Theorem 2.1.15]{nesterov2003introductory}. See also the  recent book of~\citet{nesterov2018lectures}.

To accelerate these rates of convergence, Nesterov introduced iterations of three sequences, parametrized by $\tau_k, \tau_k',\gamma_k, \gamma_k',  k\geq 0$, of the form
\begin{align}
&y_k = x_k + \tau_k(z_k-x_k) \, ,  \label{eq:nest-1}\\
&x_{k+1} = y_k - \gamma_k \nabla f(y_k) \, , \\
&z_{k+1} = z_k + \tau_k'(y_k-z_k) - \gamma_k' \nabla f(y_k ) \label{eq:nest-3}\, . 
\end{align}
Depending on whether the function $f$ is known to be (\ref{it:nest-cvx}) simply convex, or (\ref{it:nest-str-cvx}) strongly convex with a known strong convexity parameter, Nesterov gave choices of parameters leading to accelerated convergence rates.

\begin{thm}[Convergence of accelerated gradient descent]
	\label{thm:nesterov}
	\begin{enumerate}
		\item\label{it:nest-cvx} Choose the parameters $\tau_k = 1 - \frac{A_k}{A_{k+1}}, \tau_k' = 0, \gamma_k = \frac{1}{L}, \gamma_k' = \frac{A_{k+1}-A_k}{L}, k\geq 0$, where the sequence $A_k, k\geq 0$, is defined by the recurrence relation
		\begin{align*}
		&A_0 = 0 \, , &&A_{k+1} = A_k + \frac{1}{2}(1+\sqrt{4 A_k + 1}) \, .
		\end{align*}
		Then
		\begin{align*}
		f(x_k) - f(x_*) \leq \frac{2 L \Vert x_0 - x_* \Vert^2}{k^2} \, .
		\end{align*}
		\item\label{it:nest-str-cvx} Assume further that $f$ is $\mu$-strongly convex, $\mu > 0$. Choose the constant parameters \\$\tau_k \equiv \frac{\sqrt{\mu/L}}{1+\sqrt{\mu/L}}$, $\tau_k'\equiv \sqrt{\frac{\mu}{L}}$, $\gamma_k \equiv \frac{1}{L}$, $\gamma_k'\equiv \frac{1}{\sqrt{\mu L}}$, $k\geq 0$. Then	  
		\begin{align*}
		f(x_k) - f(x_*) \leq \left(f(x_0) - f(x_*) + \frac{\mu}{2} \Vert z_0 - x_* \Vert^2 \right) \left(1 - \sqrt{\frac{\mu}{L}}\right)^k \, .
		\end{align*}
	\end{enumerate}
\end{thm}
This result, in this exact form, is proven by \citet[Sections 4.4.1 and 4.5.3]{daspremont2021acceleration}.

From a high-level perspective, Nesterov acceleration iterates over several variables, alternating between gradient steps (always with respect to the gradient at $y_k$) and mixing steps, where the running value of a variable is replaced by a linear combination of the other variables. However, the precise way gradient and mixing steps are coupled is rather mysterious, and the success of the proof of Theorem \ref{thm:nesterov} relies heavily on the detailed structure of the iterations. In the next section, we try to gain perspective on this structure by developing a continuized version of the acceleration. 

\section{Continuized version of Nesterov acceleration} 
\label{sec:continuized}

In this section and the following ones, we use several mathematical notions related to random processes. It should be possible to understand the paper with only a heuristic understanding of these notions. The rigorous definitions are provided in Appendix~\ref{ap:toolbox}. 
\medskip

We argue that the accelerated iteration becomes more natural if we consider two variables $x_t, z_t$ indexed by a continuous time $t \geq 0$, that are continuously mixing and that take gradient steps at random times. More precisely, let $T_1, T_2, T_3, \dots \geq 0$ be random times such that $T_1, T_2-T_1, T_3-T_2, \dots$ are independent identically distributed (i.i.d.), of law exponential with rate $1$ (any constant rate would do, but we choose $1$ to make the comparison with discrete time $k$ straightforward). By convention, we choose that our stochastic processes $t \mapsto x_t, t \mapsto z_t$ are c\`adl\`ag almost surely, i.e., right continuous with well-defined left-limits $x_{t-}, z_{t-}$ (see Definition \ref{def:cadlag} in Appendix~\ref{ap:toolbox}). Our dynamics are parametrized by functions $\gamma_t, \gamma'_t, \tau_t, \tau_t'$, $t \geq 0$. At the random times $T_1, T_2, \dots$, our sequences take gradient steps 
\begin{align}
x_{T_k} = x_{T_k-} - \gamma_{T_k} \nabla f (x_{T_k-}) \, , \label{eq:gradient-1}\\
z_{T_k} = z_{T_k-} - \gamma_{T_k}' \nabla f (x_{T_k-}) \, . \label{eq:gradient-2}
\end{align}
Because of the memoryless property of the exponential distribution, in a infinitesimal time interval $[t, t+\diff t]$, the variables take gradients steps with probability $\diff t$, independently of the past. 

Between these random times, the variables mix through a linear ordinary differential equation (ODE)
\begin{align}
&\diff x_t = \eta_t (z_t - x_t) \diff t \, ,  \label{eq:continuous-part-1}\\
&\diff z_t = \eta_t'(x_t - z_t) \diff t \label{eq:continuous-part-2}\, .
\end{align}
Following the notation of stochastic calculus, we can write the process more compactly in terms of the Poisson point measure $\diff N(t) = \sum_{k\geq 0} \delta_{T_k}(\diff t)$, which has intensity the Lebesgue measure $\diff t$,
\begin{align}
\diff x_t = \eta_t (z_t - x_t) \diff t - \gamma_t \nabla f(x_t) \diff N(t) \, , \label{eq:continuized-1}\\
\diff z_t = \eta_t' (x_t - z_t) \diff t- \gamma_t' \nabla f(x_t) \diff N(t) \, .\label{eq:continuized-2}
\end{align}

Before giving convergence guarantees for such processes, let us digress quickly on why we can expect an iteration of this form to be mathematically appealing.

First, from a Markov chain indexed by a discrete time index $k$, one can associate the so-called \emph{continuized} Markov chain, indexed by a continuous time $t$, that makes transition with the same Markov kernel, but at random times, with independent exponential time intervals \citep{aldous1995reversible}. Following this terminology, we refer to our acceleration \eqref{eq:continuized-1}-\eqref{eq:continuized-2} as the continuized acceleration. The continuized Markov chain is appreciated for its continuous time parameter $t$, while keeping many properties of the original Markov chain; similarly the continuized acceleration is arguably simpler to analyze, while performing similarly to Nesterov acceleration. 

Second, it is also interesting to compare with coordinate gradient descent methods, that are easier to analyze when coordinates are selected randomly rather than in an ordered way~\citep{wright2015coordinate}. Similarly, the continuized acceleration is simpler to analyze because the gradient steps \eqref{eq:gradient-1}-\eqref{eq:gradient-2} and the mixing steps \eqref{eq:continuous-part-1}-\eqref{eq:continuous-part-2} alternate randomly, due to the randomness of $T_1, T_2, \dots$
\medskip

In analogy with Theorem \ref{thm:nesterov}, we give choices of parameters that lead to accelerated convergence rates, in the convex case \eqref{it:cvx} and in the strongly convex case \eqref{it:str-cvx}. Convergence is analyzed as a function of $t$. As $\diff N(t)$ is a Poisson point process with rate $1$, $t$ is the expected number of gradient steps done by the algorithm. Thus $t$ is analoguous to $k$ in Theorem \ref{thm:nesterov}. 
\begin{thm}[Convergence of continuized Nesterov acceleration]
	\label{thm:continuized}
	\begin{enumerate}
		\item\label{it:cvx} Choose the parameters $\eta_t = \frac{2}{t}, \eta_t' = 0, \gamma_t = \frac{1}{L}, \gamma_t' = \frac{t}{2L}$. Then 
		\begin{align*}
		\E f(x_t) - f(x_*) \leq \frac{2L\Vert z_0 -x_* \Vert^2}{t^2} \, .
		\end{align*}
		\item\label{it:str-cvx} Assume further that $f$ is $\mu$-strongly convex, $\mu > 0$. Choose the constant parameters\\ $\eta_t = \eta_t' \equiv \sqrt{\frac{\mu}{L}}$, $\gamma_t \equiv \frac{1}{L}$, $\gamma_t' \equiv \frac{1}{\sqrt{\mu L }}$. Then 
		\begin{align*}
		\E f(x_t) - f(x_*) \leq \left(f(x_0) - f(x_*) + \frac{\mu}{2} \Vert z_0 - x_* \Vert^2\right) \exp \left(-\sqrt{\frac{\mu}{L}}t\right) \, .
		\end{align*}
	\end{enumerate}
\end{thm}

\begin{proof}[Sketch]
	A complete and rigorous proof is given in Appendix \ref{ap:proof-continuized}. Here, we only provide the heuristic of the main lines of the proof. 
	
	The proof is similar to the one of Nesterov acceleration: we prove that for some choices of parameters $\eta_t, \eta_t', \gamma_t, \gamma_t'$, $t \geq 0$, and for some functions $A_t, B_t$, $t \geq 0$, 
	\begin{equation*}
	\phi_t = A_t\left(f(x_t)-f(x_*)\right) + B_t \Vert z_t - x_* \Vert^2  
	\end{equation*}
	is a supermartingale. In particular, this implies that $\E\phi_t$ is a Lyapunov function, i.e., a non-increasing function of $t$. 
	
	To prove that $\phi_t$ is a supermartingale, it is sufficient to prove that for all infinitesimal time intervals $[t,t+\diff t]$, $\E_t \phi_{t+\diff t} \leq \phi_t$, where $\E_t$ denotes the conditional expectation knowing all the past of the Poisson process up to time $t$. Thus we would like to compute the first order variation of $\E_t \phi_{t+\diff t}$. This implies computing the first order variation of $\E_t f(x_{t+\diff t})$. 
	
	From \eqref{eq:continuized-1}, we see that $f(x_t)$ evolves for two reasons between $t$ and $t + \diff t$:
	\begin{itemize}
		\item $x_t$ follows the linear ODE \eqref{eq:continuous-part-1}, which results in the infinitesimal variation $f(x_t) \rightarrow f(x_t) + \eta_t \langle \nabla f(x_t), z_t-x_t \rangle \diff t$, and
		\item with probability $\diff t$, $x_t$ takes a gradient step, which results in a macroscopic variation $f(x_t) \rightarrow f\left(x_t - \gamma_t \nabla f(x_t)\right)$. 
	\end{itemize} 
	Combining both variations, we obtain that 
	\begin{equation*}
	\E_tf(x_{t+\diff t}) \approx f(x_t) + \eta_t \langle \nabla f(x_t), z_t-x_t \rangle \diff t + \diff t \left(f\left(x_t - \gamma_t \nabla f(x_t)\right) - f(x_t)\right) \, ,
	\end{equation*}
	where the $\diff t$ in the second term corresponds to the probability that a gradient step happens; note that the latter event is independent of the past up to time $t$. 
	
	A similar computation can be done for $\E_t \Vert z_t - x_* \Vert^2$. Putting things together, we obtain 
	\begin{align*}
	\E_t \phi_{t+\diff t} - \phi_t \approx \diff t \bigg(&\frac{\diff A_t}{\diff t}(f(x_t)-f(x_*))+A_t\eta_t\langle \nabla f(x_t), z_t-x_t\rangle \\
	&- A_t \left(f(x_t-\gamma_t \nabla f(x_t))-f(x_t)\right) + \frac{\diff B_t}{\diff t} \Vert z_t - x_* \Vert^2 \\
	&+ 2B_t \eta_t' \langle z_t - x_*, x_t - z_t \rangle + B_t \big(\Vert z_t - \gamma_t'\nabla f(x_t)-x_* \Vert^2 - \Vert z_t - x_* \Vert^2 \big)\bigg) \, .
	\end{align*}
	Using convexity and strong convexity inequalities, and a few computations, we obtain the following upper bound: 
	\begin{align*}
	\E_t \phi_{t+\diff t} - \phi_t \lesssim\diff t \bigg(&\left(\frac{\diff A_t}{\diff t} - A_t \eta_t\right)  \langle \nabla f(x_t), x_t- x_* \rangle + \left(\frac{\diff B_t }{\diff t}- B_t\eta_t'\right) \Vert z_t - x_* \Vert^2 \\
	&+ (A_t \eta_t-2 B_t \gamma_t')  \langle \nabla f(x_t), z_t- x_* \rangle + \left(B_t \eta_t'- \frac{\diff A_t}{\diff t} \frac{\mu}{2}\right) \Vert x_t - x_* \Vert^2\\ 
	&  +\left( B_t \gamma_t'^2 -A_t \gamma_t \left(1 - \frac{L\gamma_t}{2}\right)\right)  \Vert \nabla f(x_t) \Vert^2 \bigg) \, .
	\end{align*}
	We want this infinitesimal variation to be non-positive. Here, we choose the parameters so that $\gamma_t = 1/L$, and all prefactors in the above expression are zero. This gives some constraints on the choices of parameters. We show that only one degree of freedom is left: the choice of the function $A_t$, that must satisfy the ODE 
	\begin{equation*}
	\frac{\diff^2}{\diff t^2} \left(\sqrt{A_t}\right) = \frac{\mu}{4L} \sqrt{A_t} \, ,
	\end{equation*}
	but whose initialization remains free. Once the initialization of the function $A_t$ is chosen, this determines the full function $A_t$ and, through the constraints, all parameters of the algorithm. As $\phi_t$ is a supermartingale (by design), a bound on the performance of the algorithm is given by
	\begin{align*}
	\E f(x_t) - f(x_*) \leq \frac{\E \phi_t}{A_t} \leq \frac{\phi_0}{A_t} \, .
	\end{align*}
	The results presented in Theorem \ref{thm:continuized} correspond to one special choice of initialization for the function~$A_t$. 
	
	In this sketch of proof, our derivation of the infinitesimal variation is intuitive and elementary; however it can be made more rigorous and concise---albeit more technical---using classical results from stochastic calculus, namely Proposition \ref{prop:ito}. This is our approach in Appendix \ref{ap:proof-continuized}.
\end{proof}

Many authors have proposed continuous-time equivalents in order to understand better Nesterov acceleration using differential calculus, see the numerous references in the introduction. For instance, in the seminal work of \cite{su2014differential}, the equivalence is obtained from Nesterov acceleration by taking the joint asymptotic where the stepsizes vanish and the number of iterates is rescaled. The resulting limit is an ODE that must be discretized to be implemented; choosing the right discretization is not straightforward as it introduces stability and approximation errors that must be controlled, see \citep{zhang2018direct,shi2019acceleration,sanz2020connections}. 

On the contrary, our continuous time equivalent \eqref{eq:continuized-1}-\eqref{eq:continuized-2} does not correspond to a limit where the stepsizes vanish. However, in Appendix \ref{ap:scaling-limit}, we check that the continuized acceleration has the same ODE scaling limit as Nesterov acceleration. This sanity check emphasizes that the continuized acceleration is fundamentally different from previous continuous-time equivalents.   


\section{Discrete implementation of the continuized implementation with random parameters}
\label{sec:implementation}

In this section, we show that the continuized acceleration can be implemented exactly as a discrete algorithm. Denote 
\begin{align*}
&\tilde{x}_k = x_{T_{k}} \, ,  &&\tilde{y}_k = x_{T_{k+1}-} \, , &&\tilde{z}_k = z_{T_{k}} \, .  
\end{align*}
The three sequences $\tilde{x}_k, \tilde{y}_k, \tilde{z}_k$, $k \geq 0$, satisfy a recurrence relation of the same structure as Nesterov acceleration, but with random weights.

\begin{thm}[Discrete version of continuized acceleration]
	\label{thm:discretization}
	For any stochastic process of the form \eqref{eq:continuized-1}-\eqref{eq:continuized-2}, we have
	\begin{align}
	&\tilde{y}_k = \tilde{x}_k + \tau_k(\tilde{z}_k-\tilde{x}_k) \, ,  \label{eq:discretization-1}\\
	&\tilde{x}_{k+1} = \tilde{y}_k - \tilde{\gamma}_k \nabla f(\tilde{y}_k) \, , \\
	&\tilde{z}_{k+1} = \tilde{z}_k + \tau_k'(\tilde{y}_k-\tilde{z}_k) - \tilde{\gamma}_k' \nabla f(\tilde{y}_k ) \, , \label{eq:discretization-3}
	\end{align}
	for some random parameters $\tau_k, \tau_k', \tilde{\gamma}_k, \tilde{\gamma}_k'$ (that are functions of $T_k, T_{k+1}, \eta_t, \eta_t', \gamma_t, \gamma_t'$).
	\begin{enumerate}
		\item For the parameters of Theorem \ref{thm:continuized}.\ref{it:cvx},
		$\tau_k = 1 - \left(\frac{T_k}{T_{k+1}}\right)^2$, $\tau_k' = 0$, $\tilde{\gamma}_k = \frac{1}{L}$, and $\tilde{\gamma}_k' = \frac{T_k}{2L}$.
		\item For the parameters of Theorem \ref{thm:continuized}.\ref{it:str-cvx}, $\tau_k = \frac{1}{2}\left(1 - \exp\left(-2\sqrt{\frac{\mu}{L}}(T_{k+1}-T_k)\right)\right)$, \\$\tau_k' = \tanh\left(\sqrt{\frac{\mu}{L}}(T_{k+1}-T_k)\right)$, $\tilde{\gamma}_k = \frac{1}{L}$, and $\tilde{\gamma}_k' = \frac{1}{\sqrt{\mu L}}$.		
	\end{enumerate}
\end{thm} 
This theorem is proved in Appendix \ref{ap:proof-thm-discretization}.  
\smallskip 

\begin{figure}
	\includegraphics[width=0.49\linewidth]{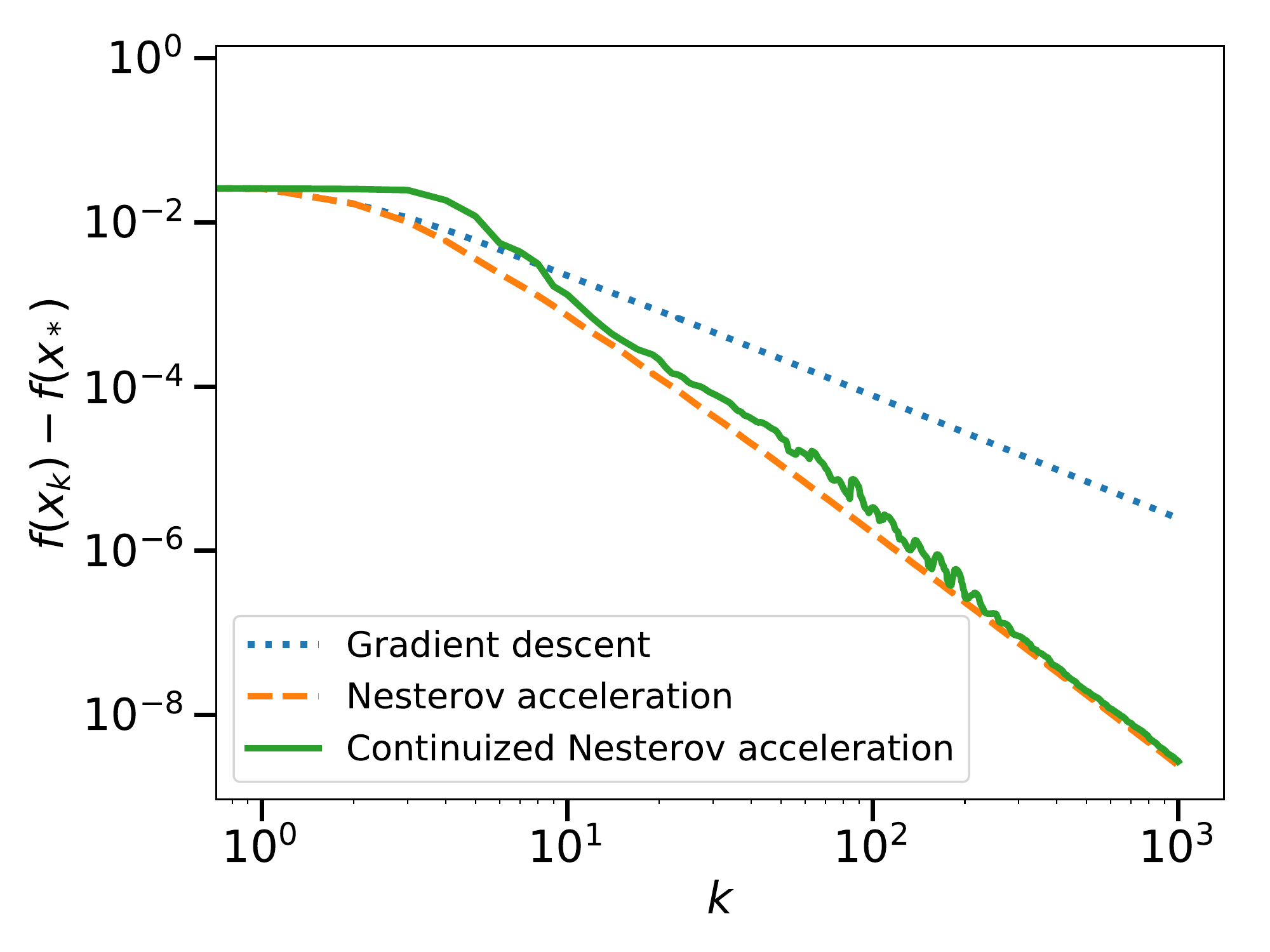}
	\includegraphics[width=0.49\linewidth]{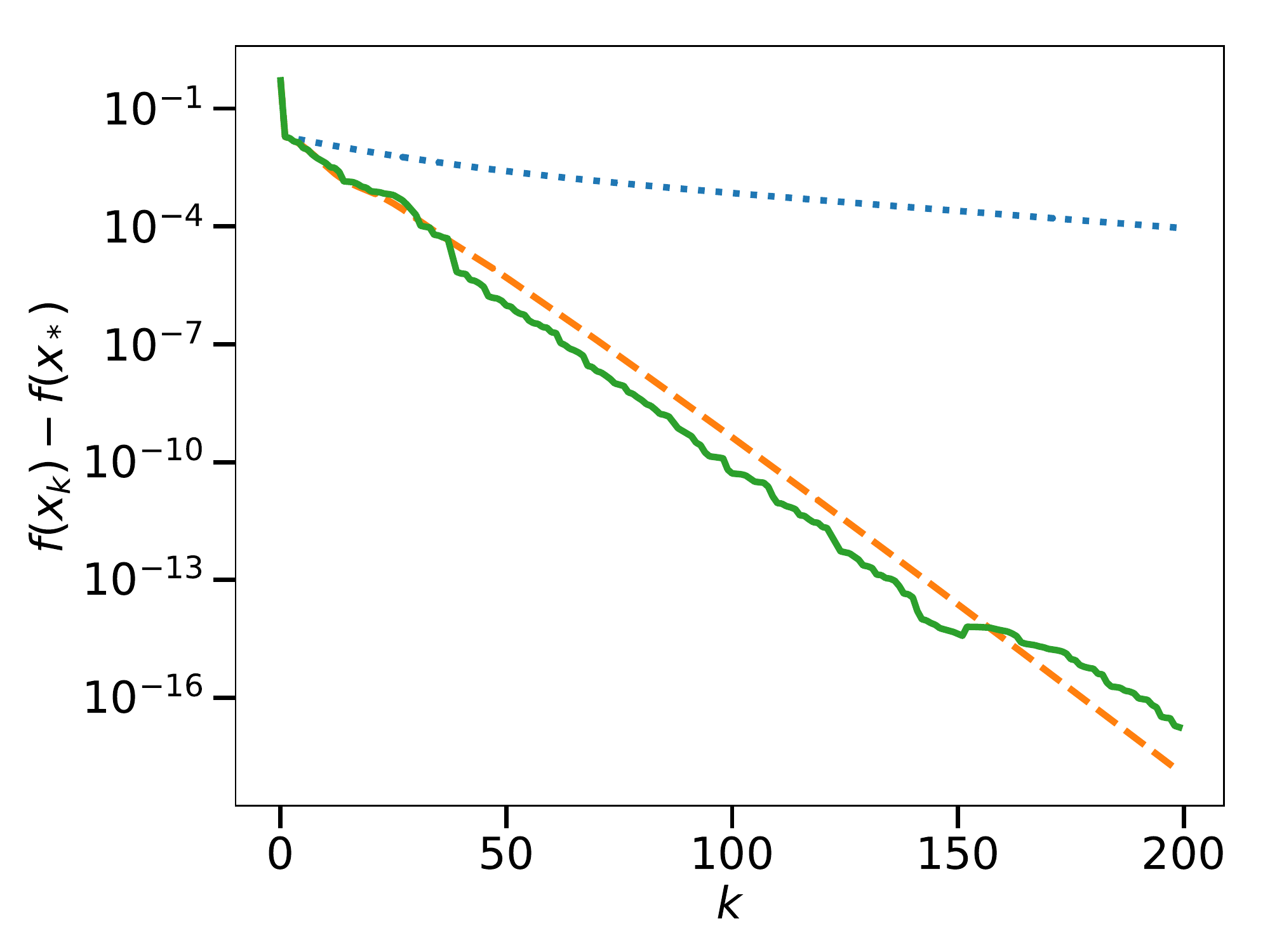}
	\caption{Comparison between gradient descent, Nesterov acceleration, and the continuized version of Nesterov acceleration, on a convex function (left) and a strongly convex function (right). For the continuized acceleration, which is randomized, the results shown corresponds to a single run. (Results were stable across runs.) }
	\label{fig:comparison}
\end{figure}


In Figure \ref{fig:comparison}, we compare this continuized Nesterov acceleration \eqref{eq:discretization-1}-\eqref{eq:discretization-3} with the classical Nesterov acceleration \eqref{eq:nest-1}-\eqref{eq:nest-3} and gradient descent. In the strongly convex case (right), we run the algorithms with the parameters of Theorem \ref{thm:nesterov}.\ref{it:nest-str-cvx} and \ref{thm:discretization}.\ref{it:str-cvx} on the function 
\begin{equation*}
f(x_1, x_2, x_3) = \frac{\mu}{2}(x_1-1)^2 + \frac{3\mu}{2}(x_2-1)^2 + \frac{L}{2}(x_3-1)^2 \, , 
\end{equation*} with $\mu = 10^{-2}$ and $L = 1$. In the convex case, we run the algorithms with the parameters of Theorem \ref{thm:nesterov}.\ref{it:nest-cvx} and \ref{thm:discretization}.\ref{it:cvx} on the function 
\begin{equation*}
f(x_1, \dots, x_{100})= \frac{1}{2}\sum_{i = 1}^{100} \frac{1}{i^2}\left(x_i - \frac{1}{i}\right)^2 \, , 
\end{equation*}
which has negligible strong convexity parameter. All iterations were initialized from $x_0 = z_0 = 0$.  
\smallskip

In order to have a straightforward theoretical comparison with Nesterov acceleration, we describe the performance $f(\tilde{x}_{k}) - f(x_*) = f(x_{T_k}) - f(x_*)$ of the continuized acceleration in terms of the number $k$ of gradient operations. 

\begin{thm}[Convergence of the discretized version]
	\label{thm:continuized-discrete}
	The discrete implementation \eqref{eq:discretization-1}-\eqref{eq:discretization-3}, with random weights, of the continuized acceleration, satisfies: 
	\begin{enumerate}
		\item For the parameters of Theorem \ref{thm:discretization}.\ref{it:cvx},
		\begin{align*}
		\E \left[T_k^2\left(f(\tilde{x}_{k}) - f(x_*)\right) \right] \leq {2L\Vert z_0 -x_* \Vert^2} \, .
		\end{align*}
		\item Assume further that $f$ is $\mu$-strongly convex, $\mu > 0$. For the parameters of Theorem \ref{thm:discretization}.\ref{it:str-cvx}, \begin{align*}
		\E \left[\exp \left(\sqrt{\frac{\mu}{L}}T_k\right)\left(f(\tilde{x}_{k}) - f(x_*)\right) \right] \leq f(x_0) - f(x_*) + \frac{\mu}{2} \Vert z_0 - x_* \Vert^2  \, .
		\end{align*}	
	\end{enumerate}
\end{thm}
This theorem is proved in Appendix \ref{ap:proof-continuized}. The law of $T_k$ is well known: it is the sum of $k$ i.i.d.~random variables of law exponential with rate $1$; this is called an Erlang or Gamma distribution with shape parameter $k$ and rate $1$. One can use well-known properties of this law, such as its concentration around its expectation $\E T_k = k$, to derive corollaries of Theorem \ref{thm:continuized-discrete}.

\section{Robustness of the continuized Nesterov acceleration to additive noise}
\label{sec:robustness}

We now investigate how the continuized version of Nesterov acceleration performs under stochastic noise. We should emphasize that a similar study has been done on Nesterov acceleration directly \citep{lan2012optimal,hu2009accelerated,xia2010dual,devolder2011stochastic,cohen2018acceleration,aybat2020robust}. However, in the continuized framework, the randomness of the stochastic gradient and its time mix in a particularly convenient way. 

We assume that we do not have direct access to the gradient $\nabla f(x)$ but to a random estimate $\nabla f(x, \xi)$, where $\xi \in \Xi$ is random of law $\cP$. We assume that our estimate is unbiased, i.e.,
\begin{equation}
\label{eq:unbiased}
\forall x \in \R^d \, , \qquad \E_\xi \nabla f(x, \xi) = \nabla f(x) \, ,
\end{equation}
and has a uniformly bounded variance, i.e., there exists $\sigma^2 \geq 0$ such that 
\begin{equation}
\label{eq:bounded-variance}
\forall x \in \R^d \, , \qquad \E_\xi \left\Vert \nabla f(x, \xi) - \nabla f(x) \right\Vert^2 \leq \sigma^2 \, .
\end{equation}
These assumptions typically hold in the additive noise model, where $\nabla f(x,\xi) = \nabla f(x) + \xi$, where $\xi \in \R^d$ is satisfies $\E\xi = 0$, $\E \Vert \xi \Vert^2 \leq \sigma^2$. By an abuse of terminology, we say that our stochastic gradients have ``additive noise'' when \eqref{eq:unbiased} and \eqref{eq:bounded-variance} hold. 

\medskip
We keep the same algorithms, replacing gradients by stochastic gradients. Let $\xi_1, \xi_2, \dots$ be i.i.d.~random variables of law $\cP$. We take stochastic gradient steps at the random times $T_1, T_2, \dots$, 
\begin{align*}
x_{T_k} = x_{T_k-} - \gamma_{T_k} \nabla f (x_{T_k-}, \xi_k) \, , \\
z_{T_k} = z_{T_k-} - \gamma_{T_k}' \nabla f (x_{T_k-}, \xi_k) \, . 
\end{align*}
Between these random times, the variables mix through the same ODE
\begin{align*}
&\diff x_t = \eta_t (z_t - x_t) \diff t \, ,  \\
&\diff z_t = \eta_t'(x_t - z_t) \diff t \, .
\end{align*}
This can be written more compactly in terms of the Poisson point measure $\diff N(t,\xi) = \sum_{k\geq 0} \delta_{(T_k,\xi_k)}(\diff t, \diff \xi)$ on $\R_{\geq 0} \times \Xi$, which has intensity $\diff t \otimes \cP$,
\begin{align}
\diff x_t = \eta_t (z_t - x_t) \diff t - \gamma_t \int_{\Xi}\nabla f(x_t, \xi) \diff N(t,\xi) \, , \label{eq:continuized-sgd-additive-1}\\
\diff z_t = \eta_t' (x_t - z_t) \diff t- \gamma_t'  \int_{\Xi}\nabla f(x_t, \xi) \diff N(t,\xi) \, . \label{eq:continuized-sgd-additive-2}
\end{align}

\begin{thm}[Continuized acceleration with noise]
	\label{thm:additive-noise}
	Assume that the stochastic gradients are unbiased \eqref{eq:unbiased} and have a variance uniformly bounded by $\sigma^2$ \eqref{eq:bounded-variance}. Then the continuized acceleration \eqref{eq:continuized-sgd-additive-1}-\eqref{eq:continuized-sgd-additive-2} satisfies the following. 
	\begin{enumerate}
		\item For the parameters of Theorem \ref{thm:continuized}.\ref{it:cvx},
		\begin{align*}
		\E f(x_t) - f(x_*) \leq \frac{2L\Vert z_0 -x_* \Vert^2}{t^2} + \sigma^2 \frac{t}{3L} \, .
		\end{align*}
		\item Assume further that $f$ is $\mu$-strongly convex, $\mu > 0$. For the parameters of Theorem \ref{thm:continuized}.\ref{it:str-cvx}, 
		\begin{align*}
		\E f(x_t) - f(x_*) \leq \left(f(x_0) - f(x_*) + \frac{\mu}{2} \Vert z_0 - x_* \Vert^2\right) \exp \left(-\sqrt{\frac{\mu}{L}}t\right) + \sigma^2 \frac{1}{\sqrt{\mu L}} \, .
		\end{align*}	
	\end{enumerate}
\end{thm}
This theorem is proved in Appendix \ref{ap:proof-additive}.  

In the above bounds, $L$ is a parameter of the algorithm, that can be taken greater than the best known smoothness constant of the function $f$. Increasing $L$ reduces the stepsizes of the algorithm and performs some variance reduction. If the bound $\sigma^2$ on the variance is known, one can choose $L$ optimizing the above bounds in order to obtain algorithms that adapt to additive noise. 

\begin{figure}
	\includegraphics[width=0.49\linewidth]{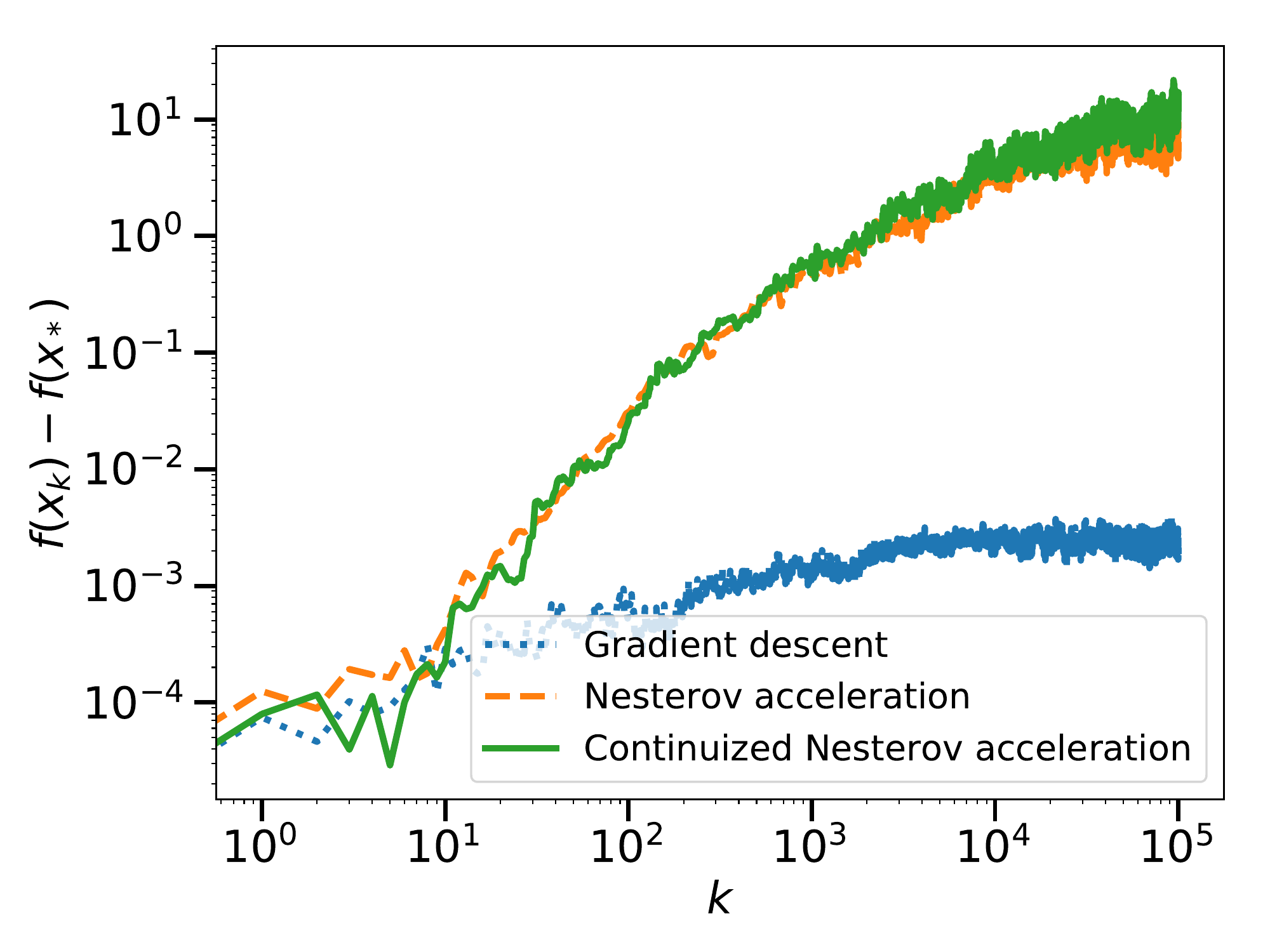}
	\includegraphics[width=0.49\linewidth]{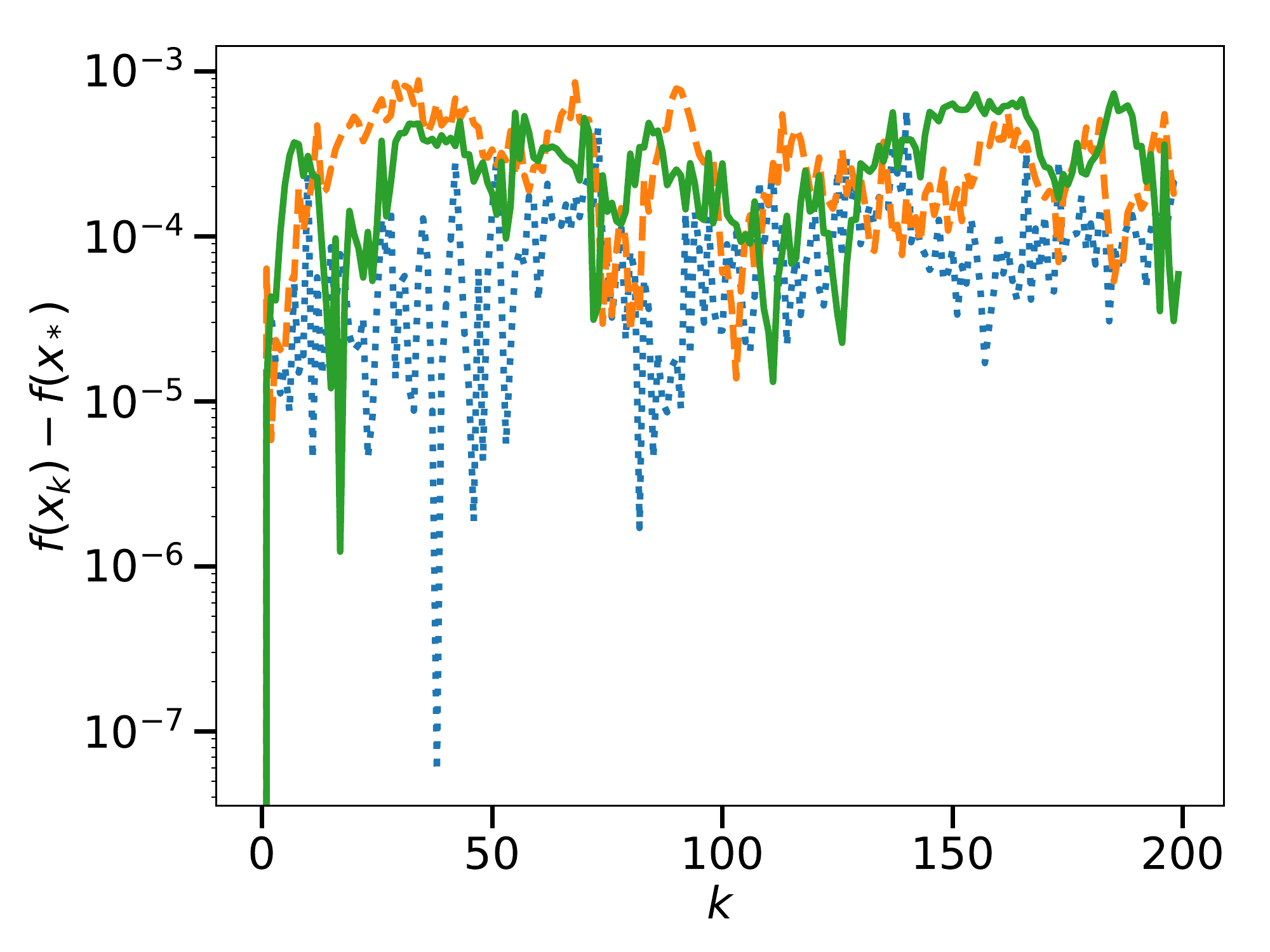}
	\caption{Effect of additive noise on gradient descent, Nesterov acceleration, and the continuized version of Nesterov acceleration, on a convex function (left) and a strongly convex function (right). The results shown corresponds to a single run. (Results were stable across runs.) }
	\label{fig:comparison-add}
\end{figure}
In Figure \ref{fig:comparison-add}, we run the same simulations as in Figure \ref{fig:comparison}, with two differences: (1) we add isotropic Gaussian noise on the gradients, with covariance $10^{-4}\Id$, and (2) we initialized algorithms at the optimum, i.e., $x_0 = z_0 = x_*$. Initializing at the optimum enables to isolate the effect of the additive noise only. These simulations confirm Theorem \ref{thm:additive-noise}: the noise term is (sub-)linearly increasing in the convex case and constant in the strongly convex case. 

Note that similarly to Theorem \ref{thm:continuized-discrete}, one could obtain convergence bounds for the discrete implementation under the presence of additive noise.

\section{Conclusion}
In this work, we introduced a continuized version of Nesterov's accelerated gradients. In a nutshell, the method has two sequences of iterates from which gradient steps are taken at random times. In between gradient steps, the two sequences mix following a simple ordinary differential equation, whose parameters are picked for ensuring good convergence properties of the method.

As compared to other continuous time models of Nesterov acceleration, a key feature of this approach is that the method can be implemented without any approximation step, as the differential equation governing the mixing procedure has a simple analytical solution. When discretized, the continuized method corresponds to an accelerated gradient method with random parameters.  

Continuization strategies were introduced in the context of Markov chains~\citep{aldous1995reversible}. Here, they allow using acceleration mechanisms in asynchronous distributed optimization, where agents are usually not aware of total the number of iterations taken so far, as showcased in the context of asynchronous gossip algorithms by~\citet{even2020asynchrony}. Possible future research directions include extending to constrained and non-Euclidean settings. 

\section*{Acknowledgements}

This work was funded in part by the French government under management of Agence Nationale de la Recherche as part of the “Investissements d’avenir” program, reference ANR-19-P3IA-0001 (PRAIRIE 3IA Institute). We also acknowledge support from the European Research Council (grant SEQUOIA 724063) and from the DGA. 

\bibliographystyle{apalike}
\bibliography{bibliography}

\appendix

\newpage

\section{Stochastic calculus toolbox}
\label{ap:toolbox}

In this appendix, we give a short introduction to the mathematical tools that we use in this paper. For more details, the reader can consult the more rigorous monographs of \cite{jacod2013limit,ikeda2014stochastic,le2016brownian}.

\subsection{Poisson point measures}

We fix $\cP$ a probability law on some space $\Xi$. 

\begin{definition}
	A \emph{(homogenous) Poisson point measure} on $\R_{\geq 0} \times \Xi$, with intensity $\nu(\diff t, \diff \xi) = \diff t \otimes \diff \cP(\xi)$, is a random measure $N$ on $\R_{\geq 0} \times \Xi$ such that
	\begin{itemize}
		\item For any disjoint measurable subsets $A$ and $B$ of $\R_{\geq 0} \times \Xi$, $N(A)$ and $N(B)$ are independent.
		\item For any measurable subset $A$ of $\R_{\geq 0} \times \Xi$, $N(A)$ is a Poisson random variable with parameter $\nu(A)$. (If $\nu(A) = \infty$, $N(A)$ is equal to $\infty$ almost surely.)  
	\end{itemize}
\end{definition}

\begin{proposition}
	\label{prop:decomposition-poisson-measure}
	Let $N$ be a Poisson point measure on $\R_{\geq 0} \times \Xi$ with intensity $\diff t \otimes \diff \cP(\xi)$. 
	
	There exists a decomposition $\diff N(t,\xi) = \sum_{k\geq 0} \delta_{(T_k,\xi_k)}(\diff t, \diff \xi)$ on $\R_{\geq 0} \times \Xi$ where $0 < T_1 < T_2 < T_3 < \dots$ and $\xi_1, \xi_2, \xi_3, \dots \in \Xi$ satisfy:
	\begin{itemize}
		\item $T_1, T_2-T_1, T_3 - T_2, \dots$ are i.i.d.~of law exponential with rate $1$,
		\item $\xi_1, \xi_2, \xi_3, \dots$ are i.i.d.~of law $\cP$ and independent of the $T_1, T_2, T_3, \dots$. 
	\end{itemize}
\end{proposition}

\begin{definition}
	Let $N$ be a Poisson point measure on $\R_{\geq 0} \times \Xi$ with intensity $\diff t \otimes \diff \cP(\xi)$. The \emph{filtration $\cF_t$, $t \geq 0$, generated by $N$} is defined by the formula
	\begin{align*}
	\cF_t = \sigma\left( N([0,s]\times A) \, , \, s\leq t, A \subset \Xi \text{ measurable}\right) \, .
	\end{align*}
\end{definition}

\subsection{Martingales and supermartingales}

Let $(\Omega, \cF, \P)$ be a probability space and $\cF_t$, $t \geq 0$, a filtration on this probability space.

\begin{definition}
	A random process $x_t \in \R^d$, $t \geq 0$, is \emph{adapted} if for all $t \geq 0$, $x_t$ is $\cF_t$-measurable. 
	An adapted process $x_t \in \R$, $t \geq 0$ is a \emph{martingale} (resp.~\emph{supermartingale}) if for all $0 \leq s \leq t$, $\E[x_t | \cF_s] = x_s$ (resp.~$\E[x_t | \cF_s] \leq x_s$).
\end{definition}

\begin{definition}
	A random variable $T \in [0, \infty]$ is a \emph{stopping time} if for all $t \geq 0$, $\{T \leq t\} \in \cF_t$. 
\end{definition}

\begin{definition}
	\label{def:cadlag}
	A function $x_t, t \geq 0$, is said to be \emph{c\`adl\`ag} if it is right continuous and for every $t > 0$, the limit $x_{t-} := \lim_{s \to t, s < t} x_s$ exists and is finite. 
\end{definition}

\begin{thm}[Martingale stopping theorem]
	\label{thm:stopping}
	Let $x_t$, $ t \geq 0$, be a martingale (resp.~supermartingale) with c\`adl\`ag trajectories and uniformly integrable. Let $T$ be a stopping time. Then $\E X_T = X_0$ (resp.~$\E X_T \leq X_0$).
\end{thm}

\subsection{Stochastic ordinary differential equation with Poisson jumps}

We fix $\cP$ a probability law on some space $\Xi$, $N$ a Poisson point measure on $\R_{\geq 0} \times \Xi$ with intensity $\diff t \otimes \diff \cP(\xi)$, and denote $\cF_t$, $t \geq 0$, the filtration generated by $N$. 

\begin{definition}
	Let $b: \R^d \to \R^d$ and $G: \R^d \times \Xi \to \R^d$ be two functions. An random process $x_t \in \R^d$, $t \geq 0$, is said to be a solution of the equation
	\begin{align*}
	\diff x_t = b(x_t) \diff t + \int_{\Xi} G(x_t, \xi) \diff N(t, \xi)
	\end{align*}
	if it is adapted, c\`adl\`ag, and for all $t \geq 0$,
	\begin{align*}
	x_t =  x_0 + \int_{0}^{t}b(x_s) \diff s + \int_{[0,t]\times \Xi} G(x_{s-}, \xi) \diff N(s, \xi) \, .
	\end{align*}
	If we consider the decomposition $\diff N(t,\xi) = \sum_{k\geq 0} \delta_{(T_k,\xi_k)}(\diff t, \diff \xi)$ given by Proposition \ref{prop:decomposition-poisson-measure}, then 
	\begin{align*}
	\int_{[0,t]\times \Xi} G(x_{s-}, \xi) \diff N(s, \xi) = \sum_{k \geq 0} \bfone_{\{T_k \leq t\}} G(x_{T_k-}, \xi_k) \, .
	\end{align*}
\end{definition}

\begin{proposition}
	\label{prop:ito}
	Let $x_t \in \R^d$ be a solution of
	\begin{align*}
	\diff x_t = b(x_t) \diff t + \int_{\Xi} G(x_t, \xi) \diff N(t, \xi)
	\end{align*}
	and $\varphi : \R^d \to \R$ be a smooth function. Then 
	\begin{align*}
	\varphi(x_t) = \varphi(x_0) + \int_{0}^{t} \langle \nabla \varphi(x_s), b(x_s) \rangle \diff s + \int_{[0,t]\times \Xi}\left(\varphi(x_{s-}+G(x_{s-}, \xi) ) - \varphi(x_{s-}) \right) \diff N(s, \xi) \, . 
	\end{align*}
	Moreover, we have the decomposition 
	\begin{align*}
	&\int_{[0,t]\times \Xi} \left(\varphi(x_{s-}+G(x_{s-}, \xi) ) - \varphi(x_{s-}) \right) \diff N(s, \xi)  \\
	&\qquad = 	\int_{0}^{t}\int_{\Xi} \left(\varphi(x_{s}+G(x_{s}, \xi) ) - \varphi(x_{s} )\right) \diff t \diff \cP(\xi) + M_t \, , 
	\end{align*}
	where $M_t = \int_{[0,t]\times \Xi} \left(\varphi(x_{s-}+G(x_{s-}, \xi) ) - \varphi(x_{s-}) \right) (\diff N(s, \xi) - \diff t \diff \cP(\xi))$ is a martingale. 
\end{proposition}

This proposition is an elementary calculus of variations formula: to compute the value of the observable $\varphi(x_t)$, one must sum the effects of the continuous part and of the Poisson jumps. Moreover, the integral with respect to the Poisson measure $N$ becomes a martingale if the same integral with respect to its intensity measure $\diff t \otimes \diff \cP(\xi)$ is removed.

\section{Analysis of the continuized Nesterov acceleration}
\label{ap:analysis-continuized}

To encompass the proofs in the convex and in the strongly convex cases in a unified way, we assume $f$ is $\mu$-strongly convex, $\mu \geq 0$. If $\mu > 0$, this corresponds to assuming the $\mu$-strong convexity in the usual sense; if $\mu = 0$, it means that we only assume the function to be convex. In other words, the proofs in the convex case can be obtained by taking $\mu = 0$ below.   

In this section, $\cF_t$, $t\geq 0$, is the filtration associated to the Poisson point measure $N$.

\subsection{Noiseless case: proofs of Theorems \ref{thm:continuized} and \ref{thm:continuized-discrete}}
\label{ap:proof-continuized}

In this section, we analyze the convergence of the continuized iteration \eqref{eq:continuized-1}-\eqref{eq:continuized-2}, that we recall for the reader's convenience:
\begin{align*}
\diff x_{t} &= \eta_t (z_t-x_t) \diff t - \gamma_t \nabla f(x_t)\diff N(t)  \, ,\\
\diff z_{t} &= \eta_t' (x_t-z_t)\diff t - \gamma_t'\nabla f(x_t)\diff N(t) \, .
\end{align*}
The choices of parameters $\eta_t, \eta'_t, \gamma_t, \gamma'_t$, $t\geq 0$, and the corresponding convergence bounds follow naturally from the analysis. We seek sufficient conditions under which the function
\begin{equation*}
\phi_t = A_t\left(f(x_t)-f_*\right) + B_t \Vert z_t - x_* \Vert^2  
\end{equation*}
is a supermartingale. 

The process $\bar{x}_t = (t,x_t,z_t)$ satisfies the equation 
\begin{align*}
&\diff \bar{x}_t = b(\bar{x}_t) \diff t + G(\bar{x}_t) \diff N(t) \, , && b(\bar{x}_t) = \begin{pmatrix}
1 \\
\eta_t(z_t-x_t) \\
\eta_t'(x_t-z_t) 
\end{pmatrix} \, , &&G(\bar{x}_t) = \begin{pmatrix}
0 \\
-\gamma_t \nabla f(x_t) \\
-\gamma_t' \nabla f(x_t)
\end{pmatrix} \, .
\end{align*}
We thus apply Proposition \ref{prop:ito} to $\phi_t = \varphi(\bar{x}_t) = \varphi(t,x_t,z_t)$ where 
\begin{align*}
\varphi(t,x,z) = A_t\left(f(x)-f(x_*)\right) + B_t\Vert z - x_* \Vert^2 \, ,
\end{align*}
we obtain:
\begin{align*}
\phi_t = \phi_0 + \int_{0}^{t} \langle \nabla \varphi(\bar{x}_s), b(\bar{x}_s) \rangle \diff s + \int_{0}^{t} \left(\varphi(\bar{x}_{s}+G(\bar{x}_{s}) ) - \varphi(\bar{x}_{s}) \right) \diff s + M_t \, ,
\end{align*}
where $M_t$ is a martingale. Thus, to show that $\varphi_t$ is a supermartingale, it is sufficient to show that the map $t \mapsto \int_{0}^{t} \langle \nabla \varphi(\bar{x}_s), b(\bar{x}_s) \rangle \diff s + \int_{0}^{t} \left(\varphi(\bar{x}_{s}+G(\bar{x}_{s}) ) - \varphi(\bar{x}_{s}) )\right) \diff s$ is non-increasing almost surely, i.e., 
\begin{align*}
I_t := \langle \nabla \varphi(\bar{x}_t), b(\bar{x}_t) \rangle +\varphi(\bar{x}_{t}+G(\bar{x}_{t}) ) - \varphi(\bar{x}_{t})  \leq 0 \, . 
\end{align*}
We now compute 
\begin{align*}
\langle \nabla \varphi(\bar{x}_{t}), b(\bar{x}_{t}) \rangle &= \partial_t \varphi(\bar{x}_{t}) + \langle \partial_x \varphi(\bar{x}_{t}), \eta_t(z_t-x_t) \rangle + \langle \partial_z \varphi(\bar{x}_{t}), \eta_t'(x_t-z_t) \rangle \\
&= \frac{\diff A_t}{\diff t} \left(f(x_t)-f(x_*)\right) + \frac{\diff B_t}{\diff t} \Vert z_t - x_* \Vert^2 + A_t \eta_t \langle \nabla f(x_t), z_t - x_t \rangle \\
&\qquad +2 B_t \eta_t' \langle z_t-x_*, x_t-z_t \rangle \, .
\end{align*}
Here, we use that as $f$ is $\mu$-strongly convex, 
\begin{equation*}
f(x_t) - f(x_*) \leq \langle \nabla f(x_t), x_t- x_* \rangle - \frac{\mu}{2} \Vert x_t - x_* \Vert^2 \, , 
\end{equation*}
and the simple bound
\begin{align*}
\langle z_t-x_*, x_t-z_t \rangle &= \langle z_t-x_*, x_t-x_* \rangle - \Vert z_t - x_* \Vert^2 \leq \Vert z_t - x_* \Vert \Vert x_t - x_* \Vert - \Vert z_t - x_* \Vert^2 \\
&\leq \frac{1}{2} \left(\Vert z_t - x_* \Vert^2  + \Vert x_t - x_* \Vert^2 \right) - \Vert z_t - x_* \Vert^2 = \frac{1}{2} \left(\Vert x_t - x_* \Vert^2 -\Vert z_t - x_* \Vert^2 \right) \, .
\end{align*}
This gives
\begin{align}
\langle \nabla \varphi(\bar{x}_{t}), b(\bar{x}_{t}) \rangle &\leq
\left(\frac{\diff A_t}{\diff t} - A_t \eta_t\right)  \langle \nabla f(x_t), x_t- x_* \rangle + \left(B_t \eta_t'- \frac{\diff A_t}{\diff t} \frac{\mu}{2}\right)  \Vert x_t - x_* \Vert^2 \label{eq:aux-4}\\&+ \left(\frac{\diff B_t }{\diff t}- B_t\eta_t'\right) \Vert z_t - x_* \Vert^2
+ A_t \eta_t \langle \nabla f(x_t), z_t - x_* \rangle \, . \label{eq:aux-5}
\end{align}
Further, 
\begin{align*}
\varphi(\bar{x}_{t}+G(\bar{x}_{t}) ) - \varphi(\bar{x}_{t}) &= A_t\left(f(x_t-\gamma_t \nabla f(x_t))-f(x_t)\right) \\
&\qquad+ B_t \left(\Vert (z_t-x_*) - \gamma_t' \nabla f(x_t) \Vert^2 - \Vert z_t-x_*\Vert^2\right) \, .
\end{align*}
As $f$ is $L$-smooth, 
\begin{align*}
f(x_t-\gamma_t \nabla f(x_t))-f(x_t) &\leq \langle \nabla f(x_t) , - \gamma_t \nabla f(x_t) \rangle + \frac{L}{2} \Vert \gamma_t \nabla f(x_t) \Vert^2 \\
&= -\gamma_t \left(1 - \frac{L\gamma_t}{2}\right) \Vert \nabla f(x_t) \Vert^2 \, .
\end{align*}
This gives
\begin{align}
\varphi(\bar{x}_{t}+G(\bar{x}_{t}) ) - \varphi(\bar{x}_{t}) &\leq
\left( B_t \gamma_t'^2 -A_t \gamma_t \left(1 - \frac{L\gamma_t}{2}\right)\right)  \Vert \nabla f(x_t) \Vert^2 -2 B_t \gamma_t' \langle \nabla f(x_t), z_t- x_* \rangle  \, . \label{eq:aux-6}
\end{align}
Finally, combining \eqref{eq:aux-4}-\eqref{eq:aux-5} with \eqref{eq:aux-6}, we obtain
\begin{align}
I_t  &\leq \left(\frac{\diff A_t}{\diff t} - A_t \eta_t\right)  \langle \nabla f(x_t), x_t- x_* \rangle + \left(\frac{\diff B_t }{\diff t}- B_t\eta_t'\right) \Vert z_t - x_* \Vert^2 \label{eq:aux-7}\\
&\qquad+ (A_t \eta_t-2 B_t \gamma_t')  \langle \nabla f(x_t), z_t- x_* \rangle + \left(B_t \eta_t'- \frac{\diff A_t}{\diff t} \frac{\mu}{2}\right) \Vert x_t - x_* \Vert^2\\ 
&\qquad   +\left( B_t \gamma_t'^2 -A_t \gamma_t \left(1 - \frac{L\gamma_t}{2}\right)\right)  \Vert \nabla f(x_t) \Vert^2 \, . \label{eq:aux-8}
\end{align}
Remember that $I_t \leq 0$ is a sufficient condition for $\phi_t$ to be a supermartingale. Here, we choose the parameters $\eta_t, \eta_t', \gamma_t, \gamma_t', t \geq 0$, so that all prefactors are $0$. We start by taking $\gamma_t \equiv \frac{1}{L}$ (other choices $\gamma_t < \frac{2}{L}$ could be possible but would give similar results) and we want to satisfy
\begin{align*}
&\frac{\diff A_t}{\diff t} = A_t \eta_t \, , &&\frac{\diff B_t }{\diff t}= B_t\eta_t' 
&&A_t \eta_t = 2 B_t \gamma_t' \, , &&B_t \eta_t' = \frac{\diff A_t }{\diff t} \frac{\mu}{2} \, , 
&&B_t \gamma_t'^2 = \frac{A_t}{2L} \, .
\end{align*}
To satisfy the last equation, we choose 
\begin{equation}
\label{eq:formula-gamma_t'}
\gamma_t' = \sqrt{\frac{A_t}{2LB_t}} \, . 
\end{equation}
To satisfy the third equation, we choose
\begin{equation}
\label{eq:formula-eta_t}
\eta_t = \frac{2 B_t \gamma_t'}{A_t} = \sqrt{\frac{2B_t}{LA_t }} \, . 
\end{equation}
To satisfy the fourth equation, we choose
\begin{equation}
\label{eq:formula-eta_t'}
\eta_t' = \frac{\diff A_t}{\diff t} \frac{\mu}{2 B_t} = \frac{A_t \eta_t \mu}{2 B_t} = \mu \sqrt{\frac{A_t}{2L B_t}} \, . 
\end{equation}
Having now all parameters $\eta_t, \eta_t', \gamma_t, \gamma_t'$ constrained, we now have that $\phi_t$ is Lyapunov if
\begin{align*}
&\frac{\diff A_t}{\diff t} = A_t \eta_t = \sqrt{\frac{2A_t B_t}{L}} \, , &&\frac{\diff B_t}{\diff t} = B_t \eta'_t = \mu \sqrt{\frac{A_t B_t}{2L}} \, .
\end{align*}
This only leaves the choice of the initialization $(A_0, B_0)$ as free: both the algorithm and the Lyapunov depend on it. (Actually, only the relative value $A_0/B_0$ matters.) Instead of solving the above system of two coupled non-linear ODEs, it is convenient to turn them into a single second-order linear ODE:
\begin{align}
\label{eq:coupled-first-order}
&\frac{\diff}{\diff t}\left(\sqrt{A_t}\right) = \frac{1}{2\sqrt{A_t}} \frac{\diff A_t}{\diff t} = \sqrt{\frac{B_t}{2L}} \, , &&\frac{\diff}{\diff t}\left(\sqrt{B_t}\right) = \frac{1}{2\sqrt{B_t}} \frac{\diff B_t}{\diff t} = \mu \sqrt{\frac{A_t}{8L}} \, . 
\end{align}
This can also be restated as
\begin{align}
\label{eq:second-order}
&\frac{\diff^2}{\diff t^2} \left(\sqrt{A_t}\right) = \frac{\mu}{4L} \sqrt{A_t} \, , &&\sqrt{B_t} = \sqrt{2L} \frac{\diff }{\diff t}\left(\sqrt{A_t}\right) \, .
\end{align}

\subsubsection{Proof of the first part (convex case)}

We now assume $\mu = 0$, and we choose the solution such that $A_0 = 
0$ and $B_0 = 1$. From \eqref{eq:coupled-first-order}, we have $\frac{\diff}{\diff t}\left(\sqrt{B_t}\right) = 0$, thus $B_t \equiv 1$, and $\frac{\diff}{\diff t}\left(\sqrt{A_t}\right) = \frac{1}{\sqrt{2L}}$, thus $\sqrt{A_t} = \frac{t}{\sqrt{2L}}$. 
The parameters of the algorithm are given by \eqref{eq:formula-gamma_t'}-\eqref{eq:formula-eta_t'}: $\eta_t = \frac{2}{t}$, $\eta_t' = 0$, $\gamma'_t = \frac{t}{\sqrt{2L}}$ (and we had chosen $\gamma_t = \frac{1}{L}$).

From the fact that $\phi_t$ is a supermartingale, we obtain that the associated algorithm satisfies 
\begin{equation*}
\E f(x_t) - f(x_*) \leq \frac{\E \phi_t}{A_t} \leq \frac{\phi_0}{A_t} = \frac{2L\Vert z_0 - x_* \Vert^2}{t^2} \, .
\end{equation*} 
This proves the first part of Theorem \ref{thm:continuized}. 

Further, one can apply martingale stopping Theorem \ref{thm:stopping} to the supermartingale $\phi_t$ with the stopping time $T_k$ to obtain 
\begin{align*}
\E \left[ A_{T_k} \left(f(\tilde{x}_k) - f(x_*) \right) \right]  = \E \left[ A_{T_k} \left( f(x_{T_k}) - f(x_*) \right) \right] \leq \E \phi_{T_k} \leq \phi_0 = \Vert z_0 - x_* \Vert^2 \, . 
\end{align*}
This proves the first part of Theorem \ref{thm:continuized-discrete}.

\subsubsection{Proof of the second part (strongly convex case)}

We now assume $\mu > 0$. We consider the solution of \eqref{eq:second-order} that is exponential:
\begin{align*}
&\sqrt{A_t} = \sqrt{A_0} \exp\left(\frac{1}{2} \sqrt{\frac{\mu}{L}}t\right) \, , && \sqrt{B_t } = \sqrt{A_0} \sqrt{\frac{\mu}{2}} \exp\left(\frac{1}{2} \sqrt{\frac{\mu}{L}}t\right) \, .
\end{align*}
The parameters of the algorithm are given by \eqref{eq:formula-gamma_t'}-\eqref{eq:formula-eta_t'}: $\eta_t = \eta_t' =  \sqrt{\frac{\mu}{L}}$, $\gamma'_t = \frac{1}{\sqrt{\mu L}}$ (and we had chosen $\gamma_t = \frac{1}{L}$). 

From the fact that $\phi_t$ is a supermartingale, we obtain that the associated algorithm satisfies 
\begin{align*}
\E f(x_t) - f(x_*) \leq \frac{\E \phi_t}{A_t} \leq \frac{\phi_0}{A_t} &= \frac{A_0 (f(x_0)-f(x_*)) + A_0 \frac{\mu}{2}\Vert z_0 - x_* \Vert^2}{A_t} \\&= \left(f(x_0)-f(x_*) + \frac{\mu}{2}\Vert z_0 - x_* \Vert^2\right) \exp\left(- \sqrt{\frac{\mu}{L}}t\right)\, .
\end{align*}
This proves the second part of Theorem \ref{thm:continuized}. Similarly to above, one can also apply the martingale stopping theorem to prove the second part of Theorem \ref{thm:continuized-discrete}. 

\begin{remark}
	In the above derivation, in both the convex and strongly convex cases, we choose a particular solution of \eqref{eq:second-order}, while several solutions are possible. In the convex case, we make the choice $A_0 = 0$ to have a succinct bound that does not depend on $f(x_0) - f(x_*)$. More importantly, in the strongly convex case, we choose the solution that satisfies the relation $\sqrt{\frac{\mu}{2}} \sqrt{A_t} = \sqrt{B_t}$, which implies that $\eta_t, \eta_t', \gamma_t'$, are constant functions of $t$, and $\eta_t = \eta_t'$. These conditions help solving in closed form the continuous part of the process 
	\begin{align*}
	&\diff x_t = \eta_t (z_t - x_t) \diff t \, ,  \\
	&\diff z_t = \eta_t'(x_t - z_t) \diff t \, ,
	\end{align*}
	which is crucial if we want to have a discrete implementation of our method (for more details, see Theorem \ref{thm:discretization} and its proof). However, in the strongly convex case, considering other solutions would be interesting, for instance to have an algorithm converging to the convex one as $\mu \to 0$. 
\end{remark}

\subsection{With additive noise: proof of Theorem \ref{thm:additive-noise}}
\label{ap:proof-additive}

The proof of this theorem is along the same lines as the proof of Theorem \ref{thm:continuized} above. Here, we only give the major differences. 

We analyze the convergence of the continuized stochastic iteration \eqref{eq:continuized-sgd-additive-1}-\eqref{eq:continuized-sgd-additive-2}, that we recall for the reader's convenience:
\begin{align*}
\diff x_t = \eta_t (z_t - x_t) \diff t - \gamma_t \int_{\Xi}\nabla f(x_t, \xi) \diff N(t,\xi) \, , \\
\diff z_t = \eta_t' (x_t - z_t) \diff t- \gamma_t'  \int_{\Xi}\nabla f(x_t, \xi) \diff N(t,\xi) \, . 
\end{align*}
In this setting, we loose the property that
\begin{equation*}
\phi_t = A_t\left(f(x_t)-f_*\right) + B_t \Vert z_t - x_* \Vert^2  
\end{equation*}
is a supermartingale. However, we bound the increase of $\phi_t$.

The process $\bar{x}_t = (t,x_t,z_t)$ satisfies the equation 
\begin{align*}
&\diff \bar{x}_t = b(\bar{x}_t) \diff t + \int_{\Xi}G(\bar{x}_t,\xi) \diff N(t,\xi), && b(\bar{x}_t) = \begin{pmatrix}
1 \\
\eta_t(z_t-x_t) \\
\eta_t'(x_t-z_t) 
\end{pmatrix}, &&G(\bar{x}_t,\xi) = \begin{pmatrix}
0 \\
-\gamma_t \nabla f(x_t,\xi) \\
-\gamma_t' \nabla f(x_t,\xi)
\end{pmatrix}.
\end{align*}
We apply Proposition \ref{prop:ito} to $\phi_t = \varphi(\bar{x}_t) = \varphi(t,x_t,z_t)$
and obtain
\begin{align}
\label{eq:decomposition-phi-additive}
\phi_t = \phi_0 + \int_{0}^{t} I_s \diff s + M_t \, ,
\end{align}
where $M_t$ is a martingale and
\begin{align*}
I_t = \langle \nabla \varphi(\bar{x}_t), b(\bar{x}_t) \rangle +\E_\xi\varphi(\bar{x}_{t}+G(\bar{x}_{t}, \xi) ) - \varphi(\bar{x}_{t}) \, . 
\end{align*}
The computation of the first term remains the same: the inequality \eqref{eq:aux-4}-\eqref{eq:aux-5} holds. The computation of the second term becomes 
\begin{align*}
\E_\xi \varphi(\bar{x}_{t}+G(\bar{x}_{t},\xi) ) - \varphi(\bar{x}_{t}) &= A_t\left(\E_\xi f(x_t-\gamma_t \nabla f(x_t,\xi))-f(x_t)\right) \\
&\qquad+ B_t \left(\E_\xi\Vert (z_t-x_*) - \gamma_t' \nabla f(x_t,\xi) \Vert^2 - \Vert z_t-x_*\Vert^2\right) \, .
\end{align*}
As $f$ is $L$-smooth, 
\begin{align*}
f(x_t-\gamma_t \nabla f(x_t,\xi))-f(x_t) &\leq \langle \nabla f(x_t) , - \gamma_t \nabla f(x_t,\xi) \rangle + \frac{L}{2} \Vert \gamma_t \nabla f(x_t,\xi) \Vert^2  \, , \\
\E_\xi f(x_t-\gamma_t \nabla f(x_t,\xi))-f(x_t) &\leq \langle \nabla f(x_t) , - \gamma_t \E_\xi\nabla f(x_t,\xi) \rangle + \frac{L}{2} \E_\xi\Vert \gamma_t \nabla f(x_t,\xi) \Vert^2  \, .
\end{align*}
Bu assumptions \eqref{eq:unbiased} and \eqref{eq:bounded-variance}, the stochastic gradient $\nabla f(x,\xi)$ is unbiased and has a variance bounded by $\sigma^2$, which implies $\E_\xi\Vert  \nabla f(x_t,\xi) \Vert^2 \leq \Vert   \nabla f(x_t) \Vert^2 + \sigma^2$. Thus 
\begin{align*}
\E_\xi f(x_t-\gamma_t \nabla f(x_t, \xi ))-f(x_t) 
&\leq -\gamma_t \left(1 - \frac{L\gamma_t}{2}\right) \Vert \nabla f(x_t) \Vert^2 + \sigma^2 \frac{L\gamma_t^2}{2 } \, .
\end{align*}
Similarly, 
\begin{align*}
\E_\xi\Vert (z_t-x_*) - \gamma_t' \nabla f(x_t,\xi) \Vert^2 - \Vert z_t-x_*\Vert^2 &= -2 \gamma_t' \langle \E_\xi \nabla f(x_t,\xi), z_t-x_* \rangle + \gamma_t'^2 \E_\xi \Vert  \nabla f(x_t,\xi) \Vert^2  \\
&\leq -2 \gamma_t' \langle  \nabla f(x_t), z_t-x_* \rangle + \gamma_t'^2  \Vert  \nabla f(x_t) \Vert^2 + \sigma^2 \gamma_t'^2 \, .
\end{align*}
This gives
\begin{align*}
\varphi(\bar{x}_{t}+G(\bar{x}_{t}) ) - \varphi(\bar{x}_{t}) &\leq
\left( B_t \gamma_t'^2 -A_t \gamma_t \left(1 - \frac{L\gamma_t}{2}\right)\right)  \Vert \nabla f(x_t) \Vert^2 -2 B_t \gamma_t' \langle \nabla f(x_t), z_t- x_* \rangle \\
&\qquad + \sigma^2 \left(A_t \frac{L\gamma_t^2}{2} + B_t \gamma_t'^2\right) \, .
\end{align*}
Combining the bounds, we obtain 
\begin{align*}
I_t  &\leq \left(\frac{\diff A_t}{\diff t} - A_t \eta_t\right)  \langle \nabla f(x_t), x_t- x_* \rangle + \left(\frac{\diff B_t }{\diff t}- B_t\eta_t'\right) \Vert z_t - x_* \Vert^2 \\
&\qquad+ (A_t \eta_t-2 B_t \gamma_t')  \langle \nabla f(x_t), z_t- x_* \rangle + \left(B_t \eta_t'- \frac{\diff A_t}{\diff t} \frac{\mu}{2}\right) \Vert x_t - x_* \Vert^2\\ 
&\qquad   +\left( B_t \gamma_t'^2 -A_t \gamma_t \left(1 - \frac{L\gamma_t}{2}\right)\right)  \Vert \nabla f(x_t) \Vert^2 + \sigma^2 \left(A_t \frac{L \gamma_t^2 }{2}+ B_t \gamma_t'^2 \right)\, ,
\end{align*}
which is an additive perturbation of the bound \eqref{eq:aux-7}-\eqref{eq:aux-8} in the noiseless case, with a perturbation proportional to $\sigma^2$. The choices of parameters of Theorem \ref{thm:continuized} cancel all first five prefactors, and satisfy $\gamma_t = \frac{1}{L}$, $A_t \frac{L \gamma_t^2 }{2} = B_t \gamma_t'^2$. We thus obtain 
\begin{align*}
I_t \leq \sigma^2 \frac{A_t}{L} \, . 
\end{align*}
This bound controls the increase of $\phi_t$. Using the decomposition \eqref{eq:decomposition-phi-additive}, we obtain
\begin{align*}
\E f(x_t) - f(x_*) &\leq \frac{\E \phi_t}{A_t} \leq \frac{\phi_0}{A_t} + \frac{\int_{0}^{t}\E I_s \diff s}{A_t}  \\
&\leq \frac{A_0 (f(x_0)-f(x_*)) + B_0 \Vert z_0 - x_* \Vert^2}{A_t} + \frac{\sigma^2}{L}\frac{\int_{0}^{t} A_s \diff s}{A_t} \, .
\end{align*} 

\subsubsection{Proof of the first part (convex case)}

In this case, $A_t = \frac{t^2}{2L}$ and $B_0 = 1$. Thus $\int_{0}^{t} A_s \diff s = \frac{1}{2L} \frac{t^3}{3}$. Thus
\begin{align*}
\E f(x_t) - f(x_*) \leq \frac{2L\Vert z_0 -x_* \Vert^2}{t^2} + \sigma^2 \frac{t}{3L} \, .
\end{align*} 

\subsubsection{Proof of the second part (strongly convex case)}

In this case, $A_t = A_0 \exp\left(\sqrt{\frac{\mu}{L}}t\right)$ and $B_0 = A_0 \frac{\mu}{2}$. Thus $\int_{0}^{t} A_s \diff s \leq A_0 \sqrt{\frac{\mu}{L}}^{-1}\exp\left(\sqrt{\frac{\mu}{L}}t\right) = \sqrt{\frac{L}{\mu}} A_t$. Thus 
\begin{align*}
\E f(x_t) - f(x_*) \leq \left(f(x_0) - f(x_*) + \frac{\mu}{2} \Vert z_0 - x_* \Vert^2\right) \exp \left(-\sqrt{\frac{\mu}{L}}t\right) + \sigma^2 \frac{1}{\sqrt{\mu L}} \, .
\end{align*}	

\section{Proof of Theorem \ref{thm:discretization}}
\label{ap:proof-thm-discretization}

By integrating the ODE
\begin{align*}
&\diff x_t = \eta_t (z_t - x_t) \diff t \, ,  \\
&\diff z_t = \eta_t'(x_t - z_t) \diff t \, ,
\end{align*}
between $T_k$ and $T_{k+1}-$, we obtain that there exists $\tau_k, \tau_k''$, such that 
\begin{align}
&\tilde{y}_k = x_{T_{k+1}-} = x_{T_k} + \tau_k(z_{T_k} - x_{T_k}) = \tilde{x}_k + \tau_k (\tilde{z}_k - \tilde{x}_k) \, , \label{eq:aux-1}\\
&z_{T_{k+1}-} = z_{T_k} + \tau_k''(x_{T_k} - z_{T_k}) = \tilde{z}_k + \tau_k'' (\tilde{x}_k - \tilde{z}_k) \, . \nonumber
\end{align}
From the first equation, we have $\tilde{x}_k = \frac{1}{1 - \tau_k} \left(\tilde{y}_k - \tau_k \tilde{z}_k\right)$, which gives by substitution in the second equation,
\begin{align*}
z_{T_{k+1}-} &= \tilde{z}_k + \tau_k'' \left(\frac{1}{1 - \tau_k} \left(\tilde{y}_k - \tau_k \tilde{z}_k\right) - \tilde{z}_k\right) \\
&= \tilde{z}_k + \tau_k' (\tilde{y}_k - \tilde{z}_k) \, ,
\end{align*}
where $\tau_k' = \frac{\tau_k''}{1 - \tau_k}$.

Further, from \eqref{eq:gradient-1}-\eqref{eq:gradient-2}, we obtain the equations 
\begin{align}
&\tilde{x}_{k+1} = x_{T_{k+1}} = x_{T_{k+1}-} - \gamma_{T_{k+1}} \nabla f (x_{T_{k+1}-}) = \tilde{y}_k - \gamma_{T_{k+1}} \nabla f (\tilde{y}_k) \, , \label{eq:aux-2}\\
&\tilde{z}_{k+1} = z_{T_{k+1}} = z_{T_{k+1}-} - \gamma'_{T_{k+1}} \nabla f (x_{T_{k+1}-}) = \tilde{z}_k + \tau_k' (\tilde{y}_k - \tilde{z}_k) - \gamma'_{T_{k+1}} \nabla f (\tilde{y}_k) \, .  \label{eq:aux-3}
\end{align}
The stated equation \eqref{eq:discretization-1}-\eqref{eq:discretization-3} are the combination of \eqref{eq:aux-1}, \eqref{eq:aux-2} and \eqref{eq:aux-3}. 

\begin{enumerate}
	\item The parameters of Theorem \ref{thm:continuized}.\ref{it:cvx} are $\eta_t = \frac{2}{t}, \eta_t' = 0, \gamma_t = \frac{1}{L}$ and $\gamma_t' = \frac{t}{2L}$. In this case, the ODE 
	\begin{align*}
	&\diff x_t = \eta_t (z_t - x_t) \diff t = \frac{2}{t} (z_t - x_t) \diff t \, ,  \\
	&\diff z_t = \eta_t'(x_t - z_t) \diff t = 0\, ,
	\end{align*}
	can be integrated in closed form: for $t \geq t_0$,
	\begin{align*}
	&x_t = z_{t_0} + \left(\frac{t_0}{t}\right)^2 (x_{t_0}-z_{t_0})= x_{t_0} + \left(1 - \left(\frac{t_0}{t}\right)^2\right)(z_{t_0}-x_{t_0}) \, ,  \\
	&z_t = z_{t_0}\, .
	\end{align*} 
	In particular, taking $t_0 = T_k$, $t = T_{k+1}-$, we obtain $\tau_k = 1 - \left(\frac{T_k}{T_{k+1}}\right)^2$, $\tau_k'' = 0$ and thus $\tau_k' = \frac{\tau_k''}{1 - \tau_k} = 0$. Finally, $\tilde{\gamma}_k = \gamma_{T_k} = \frac{1}{L}$ and $\tilde{\gamma}'_k = \gamma_{T_k}' = \frac{T_k}{2L}$. 
	\item The parameters of Theorem \ref{thm:continuized}.\ref{it:str-cvx} are $\eta_t = \eta_t' \equiv \sqrt{\frac{\mu}{L}}, \gamma_t \equiv \frac{1}{L}$ and $\gamma_t' \equiv \frac{1}{\sqrt{\mu L }}$. In this case, the ODE 
	\begin{align*}
	&\diff x_t = \eta_t (z_t - x_t) \diff t = \sqrt{\frac{\mu}{L}} (z_t - x_t) \diff t \, ,  \\
	&\diff z_t = \eta_t'(x_t - z_t) \diff t = \sqrt{\frac{\mu}{L}} (x_t - z_t) \diff t \, ,
	\end{align*}
	can also be integrated in closed form: for $t \geq t_0$,
	\begin{align*}
	x_t &= \frac{x_{t_0}+z_{t_0}}{2} + \frac{x_{t_0}-z_{t_0}}{2} \exp\left(-2\sqrt{\frac{\mu}{L}}(t-t_0)\right) \\ 
	&= x_{t_0} + \frac{1}{2} \left(1 - \exp\left(-2\sqrt{\frac{\mu}{L}}(t-t_0)\right)\right)(z_{t_0} - x_{t_0}) \, , \\
	z_t &= \frac{x_{t_0}+z_{t_0}}{2} + \frac{z_{t_0}-x_{t_0}}{2} \exp\left(-2\sqrt{\frac{\mu}{L}}(t-t_0)\right) \\
	&= z_{t_0} + \frac{1}{2} \left(1 - \exp\left(-2\sqrt{\frac{\mu}{L}}(t-t_0)\right)\right)(x_{t_0} - z_{t_0}) \, . 
	\end{align*}
	In particular, taking $t_0 = T_k$, $t = T_{k+1}-$, we obtain $\tau_k = \tau_k'' = \frac{1}{2} \left(1 - \exp\left(-2\sqrt{\frac{\mu}{L}}(T_{k+1}-T_k)\right)\right)$ and thus $\tau_k' = \frac{\tau_k''}{1 - \tau_k} = \tanh\left(\sqrt{\frac{\mu}{L}}(T_{k+1}-T_k)\right)$. Finally, $\tilde{\gamma}_k = \gamma_{T_k} = \frac{1}{L}$ and $\tilde{\gamma}'_k = \gamma_{T_k}' = \frac{1}{\sqrt{\mu L}}$.
\end{enumerate}

\section{Heuristic ODE scaling limit of the continuized acceleration}
\label{ap:scaling-limit}

\subsection{Convex case}

With the choices of parameters of Theorem \ref{thm:continuized}.\ref{it:cvx}, the continuized acceleration is 
\begin{align*}
\diff x_{t} &= \frac{2}{t} (z_t-x_t) \diff t - \frac{1}{L} \nabla f(x_t)\diff N(t)  \, ,\\
\diff z_{t} &=  - \frac{t}{2L}\nabla f(x_t)\diff N(t) \, .
\end{align*}
The ODE scaling limit is obtained by taking the limit $L \to \infty$ (so that the stepsize $1/L$ vanishes) and rescaling the time $s = t/\sqrt{L}$. Some law of large number argument heuristically gives us that, as $L \to \infty$, $\diff N(t) = \diff N(\sqrt{L}s) \approx \sqrt{L} \diff s$. Thus in the limit, we obtain  
\begin{align*}
\diff x_{s} &= \frac{2}{\sqrt{L}s} (z_s-x_s) \sqrt{L}\diff s - \frac{1}{L} \nabla f(x_s)\sqrt{L} \diff s  \, ,\\
\diff z_{s} &=  - \frac{\sqrt{L}s}{2L}\nabla f(x_s)\sqrt{L}\diff s  \, .
\end{align*}
The second term of the first equation becomes negligible in the limit. Thus the equations simplify to 
\begin{align*}
\frac{\diff x_{s}}{\diff s} &= \frac{2}{s} (z_s-x_s)  \, ,\\
\frac{\diff z_{s}}{\diff s} &=  - \frac{s}{2}\nabla f(x_s) \, .
\end{align*}
Thus 
\begin{align*}
- \frac{s}{2}\nabla f(x_s) = \frac{\diff z_{s}}{\diff s} = \frac{\diff }{\diff s} \left(x_s + \frac{s}{2}\frac{\diff x_{s}}{\diff s} \right) = \frac{\diff x_s }{\diff s} + \frac{1}{2}\frac{\diff x_s }{\diff s}  + \frac{s}{2}\frac{\diff^2 x_s }{\diff s^2} \, ,
\end{align*}
and thus 
\begin{align*}
\frac{\diff^2 x_s }{\diff s^2} + \frac{3}{s}\frac{\diff x_s }{\diff s} + \nabla f(x_s) = 0 \, .
\end{align*}
This is the same limiting ODE as the one found by \cite{su2014differential} for Nesterov acceleration. 

\subsection{Strongly-convex case}

With the choices of parameters of Theorem \ref{thm:continuized}.\ref{it:str-cvx}, the continuized acceleration is 
\begin{align*}
\diff x_{t} &= \sqrt{\frac{\mu}{L}} (z_t-x_t) \diff t - \frac{1}{L} \nabla f(x_t)\diff N(t)  \, ,\\
\diff z_{t} &=  \sqrt{\frac{\mu}{L}}(x_t-z_t)\diff t - \frac{1}{\sqrt{\mu L }}\nabla f(x_t)\diff N(t) \, .
\end{align*}
Again, we take joint scaling $L \to \infty$, $s = t/\sqrt{L}$, with the approximation $\diff N(t) \approx \sqrt{L} \diff s$. We obtain 
\begin{align*}
\diff x_{s} &= \sqrt{\frac{\mu}{L}} (z_s-x_s) \sqrt{L} \diff s - \frac{1}{L} \nabla f(x_s)\sqrt{L} \diff s  \, ,\\
\diff z_{s} &=  \sqrt{\frac{\mu}{L}}(x_s-z_s)\sqrt{L} \diff s - \frac{1}{\sqrt{\mu L }}\nabla f(x_s)\sqrt{L} \diff s \, .
\end{align*}
As before, the second term of the first equation becomes negligible in the limit. Thus the equations simplify to 
\begin{align}
\frac{\diff x_{s}}{\diff s} &= \sqrt{\mu} (z_s-x_s)  \, , \label{eq:aux-10}\\
\frac{\diff z_{s}}{\diff s} &=  \sqrt{\mu}(x_s-z_s)- \frac{1}{\sqrt{\mu }}\nabla f(x_s) \label{eq:aux-11}\, .
\end{align}
From \eqref{eq:aux-10}, we have $z_s = x_s + \frac{1}{\sqrt{\mu}} \frac{\diff x_{s}}{\diff s}$, and by substitution in \eqref{eq:aux-11}, we obtain 
\begin{align*}
\frac{\diff^2 x_s }{\diff s^2} + 2 \sqrt{\mu}\frac{\diff x_s }{\diff s} + \nabla f(x_s) = 0 \, .
\end{align*}
This is the so-called ``low-resolution'' ODE for Nesterov acceleration of \cite{shi2018understanding}.

\end{document}